\documentclass[letterpaper,11pt,reqno]{amsart}

\makeatletter
\usepackage{amssymb}
\usepackage{latexsym}
\usepackage{amsbsy}
\usepackage{amsfonts}
\usepackage{hyperref}
\usepackage{caption}
\usepackage{subcaption}
\usepackage{graphicx}
\usepackage{enumerate}
\usepackage{enumitem}
\usepackage{mathtools}
\usepackage{color}

\usepackage{tikz}

\def\marginpar#1{\ignorespaces}

\newtheorem{theorem}{Theorem}[section]
\newtheorem{lemma}[theorem]{Lemma}
\newtheorem{proposition}[theorem]{Proposition}
\newtheorem{corollary}[theorem]{Corollary}

\newtheorem{assump}[theorem]{Assumption}

\numberwithin{equation}{section}
\makeatother
\begin{document}
\title[Trading under PoS]{Trading under the Proof-of-Stake Protocol \\-- a Continuous-Time Control Approach}

\author[Wenpin Tang]{{Wenpin} Tang}
\address{Department of Industrial Engineering and Operations Research, Columbia University. 
} \email{wt2319@columbia.edu}

\author[David Yao]{David D.\ Yao}
\address{Department of Industrial Engineer and Operations Research, Columbia University. 
} \email{yao@columbia.edu}

\date{\today} 
\begin{abstract}

We develop a continuous-time control approach to optimal trading in a Proof-of-Stake (PoS) blockchain,
formulated as a consumption-investment problem that 
aims to strike the optimal balance between a participant's (or agent's) utility 
from holding/trading stakes and utility from consumption.
We present solutions via dynamic programming 
and the Hamilton-Jacobi-Bellman (HJB) equations.
When the utility functions are linear or convex, we derive close-form solutions and 
show that the bang-bang strategy is optimal (i.e., always buy or sell at full capacity).
Furthermore, we bring out the explicit connection between the 
rate of return in trading/holding stakes and the participant's risk-adjusted valuation of the stakes.
In particular, we show when a participant is risk-neutral or risk-seeking, corresponding to the  risk-adjusted valuation
being a martingale or a sub-martingale,
the optimal strategy must be to either buy all the time, sell all the time, or first buy then sell,
and with both buying and selling executed at full capacity.
We also propose a risk-control version of the consumption-investment problem; 
and for a special case, the ``stake-parity'' problem, 
we show a mean-reverting strategy is optimal.
\end{abstract}

\maketitle

\textit{Key words}: Consumption-investment, Proof of Stake (PoS) protocol, cryptocurrency,
dynamic programming, HJB equations, continuous-time control, 
risk control. 

\section{Introduction}

\quad 
As a digital exchange vehicle, blockchain technology has been successfully deployed in many applications 
including cryptocurrency \cite{Naka08}, 
healthcare \cite{Do19}, supply chain \cite{CTT20}, 
electoral voting \cite{W18}, and non-fungible tokens \cite{WL21}.
A blockchain is a growing chain of accounting records, called blocks,
which are jointly maintained by participants of the system using cryptography.
%
Consider for instance Bitcoin -- 
a peer to peer decentralized payment system.
In contrast to traditional payment processing networks, Bitcoin provides a permissionless environment
in which everyone is free to participate.
At the core of Bitcoin is the consensus protocol known as {\em Proof of Work} (PoW),
in which ``miners''  compete with each other by solving a hashing puzzle so as to validate
 an ever-growing log of transactions (the ``longest chain'') to update a distributed ledger;
and the miner who solves the puzzle first receives a reward (a number of coins).
Thus, while the competition is open to all participants, 
the chance of winning is proportional to a miner's computing power.

\quad Despite its popularity, the PoW protocol has some obvious drawbacks. 
Competition among miners has led to exploding levels of energy consumption in Bitcoin mining, 
\cite{Mora18, PS21}. 
\cite{AC20, AW22, CK17} pointed out that PoW mining will lead to centralization,
violating the core tenet of decentralization.
To solve the problem of energy efficiency, 
\cite{KN12, Wood14} introduced another consensus protocol -- {\em Proof of Stake} (PoS),
%
which is a bidding mechanism to select a miner to validate the new block.
Participants who choose to join the bidding process are required to commit certain stakes (coins they own),
and the winning probability is proportional to the stakes committed. 
Hence, a participant in a PoS blockchain is a ``bidder'', and 
only the winning bidder becomes the miner who does the validation.
As yet the PoS protocol has not been as popular as PoW. However, it is catching up quickly, and
blockchain developers have strong incentives to switch from a PoW to a PoS ecosystem.
A prominent case in this direction is Ethereum 2.0,
where two parallel chains -- {\em Mainnet} (PoW) and {\em Beacon Chain} (PoS)
are expected soon to merge into one unified PoS blockchain \cite{DP22}.

\quad 
There has been an active stream of recent studies on PoS in the research literature;
and here we briefly mention several that relate closely to our study.
In \cite{Saleh21} it is shown that the PoS protocol is ``without waste''
from an economic standpoint.
Issues of stability and decentralization of the PoS protocol are examined in \cite{RS21, Tang22}. 
Specifically, it is shown in \cite{RS21} that for large owners of initial wealth in a PoS system
their shares of the total wealth will remain   
stable in the long run (i.e., proportions to the total wealth will remain constant),
 and hence the rich-get-richer phenomenon will not happen.
\cite{Tang22} further extends this to medium and small participants,
and reveals a phase transition in share stability among those different types of participants.
In \cite{RS21, TY22}, various aspects of the consumption-investment problem in PoS are examined, and 
certain conditions are identified under which a participant may have no incentive to trade with others.
This leads to the complementing question, given a participant does prefer to trade, what is the
 optimal trading strategy?


\quad 
Motivated by the above question, the objective of our study here is to
develop a continuous-time control approach to optimal trading in a PoS blockchain.
While the control (or game) approach has been proposed in previous studies \cite{BBLL20, BBLL22, LRP19},  they are all
for the PoW protocol. To the best of our knowledge,
ours is the first control model developed for optimal trading  under the PoS protocol.

\quad Here is an overview of our main results. 
We first formulate the {\em consumption-investment problem}, 
which aims to strike a balance between a participant's utility from holding/trading stakes and utility from consumption.
It takes the form of a deterministic control problem with the real-time trading strategy being the control variable. 
We start with a detailed analysis on a special case that we call the ``stake-hoarding'' problem (Proposition \ref{thm:31}),
where we bring out the possible scenario of monopoly. 
We then solve the general consumption-investment problem via dynamic programming 
and the Hamilton-Jacobi-Bellman (HJB) equations (Theorem \ref{thm:41}).

\quad When the utility functions are linear or convex, more explicit solutions can be obtained, and 
we show that the bang-bang control is optimal, i.e., always buy or sell at full capacity
(Propositions \ref{coro:class} and \ref{prop:convex}).
Along with the optimal trading strategy, we are also able to bring out the explicit connection between the 
rate of return in trading/holding stakes and the participant's risk-adjusted valuation of the stakes.
In other words, the participant's risk sensitivity is explicitly accounted for in the trading strategy.   
In particular, when a participant is risk-neutral or risk-seeking, corresponding to the  risk-adjusted valuation
being a martingale or a sub-martingale,
the optimal strategy must be either buy all the time, sell all the time, or first buy then sell
(with both buying and selling executed at full capacity).

\quad Finally, we propose a risk control version of the consumption-investment problem, 
by adding a penalty term to control the level of stake holding so as to 
reduce the level of concentration risk (Theorem \ref{thm:51}).
A special case is a ``stake-parity'' problem, where the participant's
holding is controlled at a level that tries to track the system-wide average.
We show that the ``mean-reverting'' strategy is the optimal solution to the stake-parity problem (Proposition \ref{prop:spp}).


\medskip

\quad The rest of the paper is organized as follows. 
Section \ref{sc2} details the 
formulation of the consumption-investment problem under the PoS protocol.
Section \ref{sc3} presents the optimal solution to the problem,
and Section \ref{sc:linear} focuses on the special case of linear and convex utility functions. 
Section \ref{sc5} presents extensions to risk-control objectives.  
Concluding remarks are summarized in Section \ref{sc6}.

\section{Model Formulation}
\label{sc2}


\quad This section introduces the problem of trading under the PoS protocol in continuous time, and formulate
a control model to solve the problem.
%
First, collected below are some conventions that will be used throughout this paper.
\begin{itemize}[label = {--}, itemsep = 3 pt]
\item
$\mathbb{R}$ denotes the set of real numbers, 
and $\mathbb{R}_{+}$ denotes the set of nonnegative real numbers.
\item
For $x, y \in \mathbb{R}$, 
$x \wedge y$ denotes the smaller number of $x$ and $y$; 
$x \vee y$ denotes the larger number of $x$ and $y$.
\item
The symbol $x = o(y)$ means $\frac{x}{y}$ decays towards zero as $y \to \infty$.
\item
For a random variable $X$, $\mathbb{E}(X)$ denotes the expectation of $X$.
\item
Let $\Omega$ be a subset of $\mathbb{R}$.
A function $f \in \mathcal{C}^k(\Omega)$ if it is $k$-time continuously differentiable in $\Omega$.
\item
For $f \in \mathcal{C}^1([0,T])$, $f'(t)$ denotes the derivative of $f$. 
For $f \in \mathcal{C}^1([0,T] \times \Omega)$,
$\partial_t f$ (resp. $\partial_x f$) denotes the partial derivative of $f$ with respect to $t$ (resp. $x$).
\end{itemize}

\smallskip

\quad Time is continuous, indexed by $t\in [0,T]$,  for a fixed $T > 0$ representing the length of a finite horizon. 
Let $\{N(t), \, 0 \le t \le T\}$ (with $N(0):= N$) denote the process of the total volume of stakes, which are issued 
over time by the PoS protocol,
and can either be deterministic or stochastic.
For ease of presentation,
we consider a deterministic process $N(t)$, which is increasing in time and sufficiently smooth,
with the derivative $N'(t)$ representing the instantaneous rate of ``reward'' --- additional stakes (or ``coins'') 
injected into the system specified (exogenously) by the PoS protocol.
For instance, we will consider below, as a special case, the process $N(t)$ of a polynomial form: 
\begin{equation}
\label{eq:Nal}
N_\alpha(t) = (N^{\frac{1}{\alpha}} + t)^\alpha, \qquad t \ge 0.
\end{equation}
Then,  $N'_\alpha(t) = \alpha (N^{\frac{1}{\alpha}} + t)^{\alpha-1}$, 
and $N''_{\alpha}(t) = \alpha (\alpha-1)(N^{\frac{1}{\alpha}} + t)^{\alpha-2}$, 
so the parametric family \eqref{eq:Nal} covers different rewarding schemes according to the values of $\alpha$.
\begin{itemize}[itemsep = 3 pt]
\item
For $0< \alpha < 1$, we have $N''_{\alpha}(t) < 0$ so the process $N_{\alpha}(t)$ corresponds to a decreasing reward (e.g. Bitcoin);
\item
For $\alpha = 1$, the process $N_{1}(t) = N + t$ gives a rate one constant reward (e.g. Blackcoin); 
\item
For $\alpha > 1$, we get $N''_{\alpha}(t) > 0$ and hence, the process $N_{\alpha}(t)$ amounts to an increasing reward (e.g. EOS).
\end{itemize}

\quad Let $K \ge 2$ denote the total number of participants in the system, who are indexed by $k\in [K]:= \{1,\ldots, K\}$.
For each participant $k$, let $\{X_k(t), \, 0 \le t \le T\}$ (with $X_k(0) = x_k$) denote 
the process of the number of stakes that participant $k$ holds,
with $X_k(t) \ge 0$ and $\sum_{k = 1}^K X_k(t) = N(t)$ for all $t\in [0,T]$.
In the (discrete-time) PoS protocol, 
in each round of the bidding process,
individual participants commit stakes so as to be selected to validate the block and receive a reward;
 and the winning probability is $X_k(t)/N(t)$ for participant $k$, i.e., proportional to the number of stakes committed. 
(For instance, each round in Ethereum takes about $10$ seconds, corresponding to the block-generation time \cite{BV14}.)
For our continuous-time PoS model here, 
in which the time required for each round of voting is ``infinitesimal," 
imagine there are $M$ rounds of bidding during any given time interval $[t, t + \Delta t]$. 
In each round participant $k$ gets either some stake(s) or nothing; so the average total number of stakes $k$ will get over the $M$ rounds is
(by law of large numbers when $M$ is large),
\begin{equation*}
\underbrace{\frac{X_k(t)}{N(t)}  \frac{N'(t) \Delta t}{M}}_{\tiny \mbox{average number of stakes in each round}} \times \underbrace{M}_{\tiny \mbox{number of rounds}} = \quad\frac{X_k(t)}{N(t)} N'(t) \Delta t . 
\end{equation*}
Hence, replacing $\Delta t$ by the infinitesimal $dt$, we know participant $k$ will receive (on average) $\frac{X_k(t)}{N(t)}N'(t) dt$ stakes, 
where $\frac{X_k(t)}{N(t)}$ is $k$'s winning probability, 
and $N'(t)  dt$ is the reward issued by the blockchain in $[t, t+dt]$.

\quad 
Participants are allowed to trade (buy or sell) their stakes.
Participant $k$ will buy $\nu_k(t) dt$ stakes in $[t, t+dt]$ if $\nu_k(t) > 0$,
and sell $-\nu_k(t) dt$ stakes if $\nu_k(t) < 0$.
This leads to the following dynamics of participant $k$'s stakes under trading:
\begin{equation}
\label{eq:Xnu}
X'_k(t) = \nu_{k}(t) + \frac{N'(t)}{N(t)} X_k(t) \quad \mbox{for } 0 \le t \le \tau_k  \wedge T:= \mathcal{T}_k,
\end{equation}
where $\tau_k: = \inf\{t>0: X_k(t) = 0\}$ is the first time at which the process $X_k(t)$ reaches zero.
It is reasonable to stop the trading process if a participant runs out of stakes, or gets all available stakes:
\begin{itemize}[itemsep = 3 pt]
\item
If $\mathcal{T}_k = \tau_k$, then participant $k$ liquidates all his stakes by time $\tau_k$, and $X_k(\mathcal{T}_k) = 0$;
\item
If $\mathcal{T}_k = \max_{j \ne k} \tau_j$, then participant $k$ gets all issued stakes by time $\max_{j \ne k} \tau_j$,
and hence $X_k(\mathcal{T}_k) = N(\mathcal{T}_k)$.
\end{itemize}
We set $X_k(t) = X_k(\mathcal{T}_k)$ for $t > \mathcal{T}_k$.

\quad The problem is for each participant $k$ to decide how to trade stakes with others under the PoS protocol.
%
Let $\{P(t), \, 0 \le t \le T\}$ be the price process of each (unit of) stake,
which is a stochastic process assumed to be independent of the dynamics in \eqref{eq:Xnu}.
(This assumption has appeared in recent studies (e.g., \cite{RS21}), and is somehow a reflection of the reality that  
  the crypto price tends to be affected by market shocks such as macroeconomics, geopolitics, breaking news, etc much more than by trading activities.)
Here, the price $P(t)$ of each stake is measured in terms of an underlying risk-free asset (referred to as 
``cash'' for simplicity); and
let $b_{k}(t)$ denote the (units of) risk-free asset that participant $k$ holds at time $t$,
and let $r > 0$ denote the risk-free (interest) rate.
Also note that all $K$ participants are allowed to trade stakes (with cash) only internally among themselves,
whereas each participants can only exchange cash with an external source (say, a bank). 

The decision for each participant $k$ at $t$ is hence a tuple $(\nu_k(t), b_k(t))$. 
Let $\{c_k(t), \, 0 \le t \le T\}$ be the process of consumption, or cash flow of participant $k$, which follows the dynamics below: 
 \begin{equation}
dc_k(t) = rb_k(t) dt -db_k(t) - P(t) \nu_k(t) dt,  \qquad 
0 \le t \le \mathcal{T}_k; \tag{C1}
 \end{equation}
with 
 \begin{equation}
b_k(0)  = 0, \quad b_k(t) \ge 0 \mbox{ for } 0 \le t \le \mathcal{T}_k, \quad 0 \le X_k(t) \le N(t) \mbox{ for } 0 \le t \le \mathcal{T}_k. \tag{C2}
 \end{equation}
Set $b_k(t) = b_k(\mathcal{T}_k)$ and $\nu_k(t) = 0$ for $t > \mathcal{T}_k$.

\quad 
In (C1),
if $d b_k(t) < 0$, the participant sells the risk-free asset to get cash either for buying stakes, or for consumption;
if $d b_k(t) > 0$, the participant adds more risk-free asset.
Thus, (C1) is a self-financing condition in which 
$rb_k(t)dt-db_k(t)$ is the net change (in value of the risk-free asset held) used to finance new stakes $P(t) \nu_k(t) dt$ and consumption $dc(t)$.
The requirements in (C2) are all in the spirit of disallowing shorting 
on either the risk free asset $b_k(t)$ or the stakes $X_k(t)$.
In some PoS blockchains, there is a minimum requirement for bidding (e.g. 32 ETHs for Ethereum).
In this case, we can impose a lower bound on the process $X_k(t)$, to prevent it from falling below this threshold.
The analysis will be similar.
We also require that the trading strategy be bounded:
there is $\overline{\nu}_k > 0$ such that
\begin{equation}
| \nu_k(t) | \le \overline{\nu}_k. \tag{C3}
 \end{equation}
The objective of participant $k$ is:
\begin{equation}
\label{eq:obj1}
\begin{aligned}
\sup_{\{(\nu_k(t), b_k(t))\}} & J(\nu_k, b_k):=
\mathbb{E}\left\{ \int_0^{\mathcal{T}_k}e^{-\beta_k t} \left[dc_k(t) + \ell_k(X_k(t)) dt \right]  + e^{-\beta_k \mathcal{T}_k}  \left[b_k(\mathcal{T}_k) + h_k(X_k(\mathcal{T}_k) \right] \right\}  \\
& \mbox{ subject to } \eqref{eq:Xnu}, (\mbox{C}1), (\mbox{C}2), (\mbox{C}3),
\end{aligned}
\end{equation}
where $\beta_k > 0$ is a discount factor, a parameter measuring the risk sensitivity of participant $k$; 
$\ell_k(\cdot)$ and $h_k(\cdot)$ are two utility functions representing, respectively, the running profit and
the terminal profit.

\quad
While generally following Merton's consumption-investment framework,
our formulation as presented above takes into account some distinct features of PoS blockchains and cryptocurrencies.
One notable point is, the utilities $\ell$ and $h$ in the objective are expressed 
as functions of the number of stakes $X_k(t)$, as opposed to
their total value $P(t) X_k(t)$. To the extent that $P(t)$ is treated as exogenous (as explained above),
this difference may seem to be trivial.
 Yet, it is a reflection of the more substantial fact that 
crypto-participants tend to mentally decouple the utility of holding stakes from 
their monetary value at any given time.
For instance, holding $1$ ETH may be equivalent to $\$ 5,000$ for one person, and $\$ 500$ for another, 
and neither will be influenced by the ETH market price at the time, which could be say, about $\$ 1,500$.

\quad Throughout below, the following conditions will be assumed:

\smallskip
\begin{assump}~
\label{assump:1}
\begin{enumerate}[itemsep = 3 pt]
\item[(i)]
$N: [0,T] \to \mathbb{R}_{+}$ is increasing with $N(0) = N > 0$, and $N \in \mathcal{C}^2([0,T])$.
\item[(ii)]
$\ell: \mathbb{R}_{+} \to \mathbb{R}_{+}$ is increasing and $\ell \in \mathcal{C}^1(\mathbb{R}_+)$.
\item[(iii)]
$h: \mathbb{R}_{+} \to \mathbb{R}_{+}$ is increasing and $h \in \mathcal{C}^1(\mathbb{R}_+)$.
\end{enumerate}
\end{assump}


\section{The Consumption-Investment Problem}
\label{sc3}


\quad Here we study the consumption-investment problem for participant $k$ in \eqref{eq:obj1}.
To lighten notation, omit the subscript $k$, and write out the problem in full as follows,
where (C0) is a repeat of the state dynamics in \eqref{eq:Xnu}:
\begin{align}
\label{eq:obj12}
U(x):= \sup_{\{(\nu(t), b(t))\}} & J(\nu,b):= \mathbb{E}\left\{ \int_0^{\mathcal{T}}e^{-\beta t} \left[dc(t) + \ \ell(X(t)) dt\right] + e^{-\beta \mathcal{T}} \left[ b(\mathcal{T}) + h(X(\mathcal{T}) \right]\right\} \\
& \mbox{ subject to } X'(t) = \nu(t) + \frac{N'(t)}{N(t)} X(t), \, X(0) = x, \tag{C0} \\
& \qquad \qquad \quad \, dc(t) =  rb(t)dt-db(t) - P(t) \nu(t) dt , \tag{C1} \\
& \qquad \qquad \quad \,  b(0) = 0, \, b(t) \ge 0 \mbox{ and } 0 \le X(t) \le N(t), \tag{C2} \\ 
& \qquad \qquad \quad  \, |\nu(t)| \le \overline{\nu}. \tag{C3}
\end{align}
where $\mathcal{T}: = \inf\{t>0: X(t) = 0 \mbox{ or } N(t)\} \wedge T$. 

\smallskip

\quad 
Note that the expectation in the objective function is with respect to $P(t)$, which is involved
in $dc_k(t)$ via (C1). 
Denote
\begin{equation}
\label{eq:Pbeta}
\widetilde{P}_\beta(t): = e^{-\beta t} \mathbb{E} P(t) , \qquad  t\in[0, T].
\end{equation}
Substituting the constraint (C1) into the objective function, 
and taking into account 
$$rb(t) dt -db(t)   = -e^{rt}d(e^{-rt} b(t)),$$ 
along with \eqref{eq:Pbeta}, we have
\begin{eqnarray}
\label{eq:42}
J(\nu, b)
 &=&- \int_0^\mathcal{T} e^{(r - \beta)t} d(e^{-rt} b(t)) 
 + e^{-\beta \mathcal{T}} b(\mathcal{T}) \nonumber \\
&& \qquad   
+ \underbrace{\int_0^\mathcal{T} \big[-\widetilde{P}_\beta(t) \nu(t) + e^{-\beta t}\ell(X(t)\big] dt 
+ e^{-\beta \mathcal{T}}h(X(\mathcal{T}))}_{J_2(\nu)}\nonumber\\
&=& \underbrace{(r - \beta)\int_0^{\mathcal{T}} e^{-\beta t} b(t) dt}_{:=J_1(b)} \; +\; J_2( \nu) ,
\end{eqnarray}
where $b(0) = 0$ is used in the last equality.
Hence,
\begin{equation}
\label{eq:U}
U(x) := \sup_{\{(\nu, b)\}} J( \nu,b) = \sup_{b} J_1(b) + 
\sup_{\nu} J_2(\nu).
\end{equation}

\quad
Next, suppose $\beta\ge r$, a condition that will be assumed below 
(and readily justified as the risk premium associated with the valuation of any stake over the risk-free asset).
Then, from the $J_1(b)$ expression in \eqref{eq:42}, and taking into account $b(t) \ge 0$ as constrained in (C2), we have 
$\sup_b J_1(b) = 0$ with the optimality binding at $b_*(t) = 0$ for all $t$.
Consequently, the problem in \eqref{eq:obj12} is reduced to
\begin{equation}
\label{eq:45}
U(x) = \sup_{\nu} J_2(\nu) \quad \mbox{subject to (C0), (C2'), (C3)},
\end{equation}
where (C2') is (C2) without the constraints on $b(\cdot)$.

\quad In summary, the key fact here is that the objective $U(x)$ is separable 
in the control variables $(\nu(t), b(t))$; hence
 the problem in \eqref{eq:obj12} is decomposed into 
two optimal control problems, 
one on the risk-free asset $b(t)$, and the other on the trading of stakes $\nu(t)$,
as specified in \eqref{eq:42} and \eqref{eq:U}.
Moreover, under the condition $\beta \ge r$, the consumption-investment problem is reduced 
to the one in \eqref{eq:45}, where the objective function $J_2(\nu)$ -- refer to \eqref{eq:42} -- takes the form 
of a tradeoff between the utility from holding stakes ($\ell(X(t))$ and $h(X(\mathcal{T}))$)
 and the dis-utility of reducing consumption ($-\widetilde{P}_\beta(t) \nu(t)$). 
Thus, the optimal trading strategy needs to strike a balance between these two opposing terms.

\smallskip

\quad
Before we present the optimal solution to the consumption-investment problem in \eqref{eq:45},
we make a digression to first study a simple degenerate case
of $\widetilde{P}_\beta(t)\equiv 0$.
This special case removes the tradeoff mentioned above, so the solution becomes a one-sided strategy of always   
{accumulating} (or ``hoarding'') the stakes at full capacity ($\overline{\nu}$). 
Yet, as the analysis below will show, there are still some interesting (and subtle) issues involved.
More importantly, this special case provides a very accessible path to finding the optimal solution 
via dynamic programming and the HJB equation. 

\subsection{Stake-hoarding}
\label{sc3.1}


As motived above, here the problem for participant $k$ is reduced to the following (again, omit the subscript $k$):
\begin{align}
\label{eq:obj22}
U(x):=\sup_{\nu(t)} \,\, &  \int_0^{\mathcal{T}}e^{-\beta t} \ \ell(X(t)) dt  + e^{-\beta \mathcal{T}} h(X(\mathcal{T})) \\
& \mbox{ subject to  \quad (C0), (C2'), (C3)}.  \nonumber
\end{align}
Below, we denote $\nu_{*}(t)$ for the optimal control process, 
$X_{*}(t)$ for the corresponding state process,
and $\mathcal{T}_{*}:= \inf\{t>0: X_{*}(t) = N(t)\} \wedge T$ for the exit time.

\begin{proposition}
\label{thm:31}
Denote
\begin{equation}
\label{eq:gammat}
\gamma(t): = \overline{\nu} N(t) \int_0^t \frac{ds}{N(s)} + \frac{x N(t)}{N} \quad \mbox{for } 0 \le t \le T.
\end{equation}
We have:
\begin{itemize}
\item[(i)]
If $\overline{\nu} \int_0^T \frac{dt}{N(t)} \le \frac{N-x}{N}$, 
then $\mathcal{T}_{*} = T$.
The optimal control is $\nu_{*}(t) = \overline{\nu}$ for $0 \le t \le T$,
the optimal state process is $X_{*}(t) = \gamma(t)$ for $0 \le t \le T$,
and 
$U(x) = e^{-\beta T} h(X_{*}(T)) + \int_0^T e^{-\beta t} \ell(X_{*}(t)) dt$.
\item[(ii)]
If $\overline{\nu} \int_0^T \frac{dt}{N(t)} > \frac{N-x}{N}$, set
\begin{equation*}
t_0: = \inf\left\{t>0: \overline{\nu} \int_0^t \frac{ds}{N(s)} = \frac{N-x}{N}\right\} < T.
\end{equation*}
Assume further that
\begin{equation}
\label{eq:tech}
h(N(t))' + \ell(N(t)) \le \beta h(N(t)) \quad \mbox{for all } 0 \le t \le T.
\end{equation}
Then, $\mathcal{T}_{*} = t_0$.
The optimal strategy is $\nu_{*}(t) = \overline{\nu}$ for $0 \le t \le t_0$ (and $\nu_{*}(t) = 0$ for $t > t_0$),
the optimal state process is $X_{*}(t) = \gamma(t)$ for $0 \le t \le t_0$ (and $X_{*}(t) = N(t_0)$ for $t > t_0$),
and 
$U(x) = e^{-\beta t_0} h(X_{*}(t_0)) + \int_0^{t_0} e^{-\beta t} \ell(X_{*}(t)) dt$.
\end{itemize}
\end{proposition}

\begin{figure}[h]
\centering
\includegraphics[width=0.5\columnwidth]{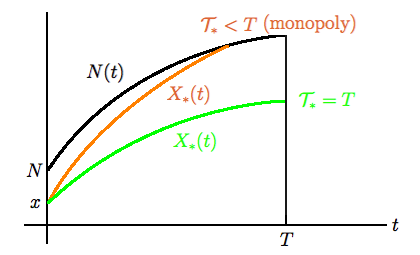}
\caption{Optimal stake trading: concentration and monopoly.}
\label{fig:1}
\end{figure}

\quad 
Deferring the proof, we first make a few comments on the above proposition.
Note that $\gamma(t)$ as specified in \eqref{eq:gammat} is identified as the optimal state process $X_*(t)$, 
which is the number of stakes given $\nu_*(t) = \overline{\nu}$.
It is easy to see that the participant's share of stakes, $X_{*}(t)/N(t)$, is increasing in $t$, 
leading to centralization regardless of how the rewarding scheme is designed
(although large rewards may slow down the speed towards concentration).
The interesting point of the above theorem is in its part (ii), where 
the required condition \eqref{eq:tech} is a technical one, to ensure the optimality of $\nu_{*}(t) = \overline{\nu}$. 
The more substantive fact is $\mathcal{T}_{*}=t_0 <T$, when $X(\mathcal{T}_*) = N(\mathcal{T}_*)$,
i.e., the participant has accumulated all stakes available in the system, leading to
the extreme situation of monopoly (or ``dictatorship''); and this is done before the end of the horizon, i.e.,
forcing a pre-matured exit time.
See Figure \ref{fig:1} for an illustration. 

\quad
The following corollary illustrates further this extreme phenomenon, 
with the polynomial family $N_{\alpha}(t)$ defined by \eqref{eq:Nal},
and with a long time horizon ($T \to \infty$).

\begin{corollary}
\label{coro:poly}
Let $(N_{\alpha}(t), \, 0 \le t \le T)$ be defined by \eqref{eq:Nal},
$(X_{\alpha, *}(t), \, 0 \le t \le T)$ be the optimal state process defined by \eqref{eq:gammat} corresponding to $N_{\alpha}(t)$,
and $\mathcal{T}_{\alpha,*}: = \inf\{t>0: X_{\alpha,*}(t) = N_{\alpha}(t)\}$ be the exit time.
Assume that the condition \eqref{eq:tech} holds for $N_{\alpha}(t)$.
Then, as $T \to \infty$, we have
\begin{itemize}
\item[(i)]
For $\alpha > 1$, 
\begin{itemize}[itemsep = 3 pt]
\item
if $\overline{\nu}  \le (\alpha - 1)(N-x) N^{-\frac{1}{\alpha}}$,
then $X_{\alpha,*}(t) < N_{\alpha}(t)$ for all $t$.
Moreover,
\begin{equation}
\label{eq:33}
\lim_{t \to \infty} \frac{X_{\alpha,*}(t)}{N_{\alpha}(t)} =  \frac{\overline{\nu}}{\alpha - 1} N^{\frac{1-\alpha}{\alpha}} + \frac{x}{N}.
\end{equation}
\item
if $\overline{\nu} > (\alpha - 1)(N-x) N^{-\frac{1}{\alpha}}$,
then $\mathcal{T}_{\alpha, *} < \infty$.
\end{itemize}
\item[(ii)]
For $\alpha \le 1$, 
we have $\mathcal{T}_{\alpha, *} < \infty$.
\end{itemize}
\end{corollary}

\begin{proof}
Note that 
\begin{equation}
\label{eq:Nal1}
\int_0^T \frac{dt}{N_{\alpha}(t)} =
\left\{ \begin{array}{lcl}
\frac{1}{1-\alpha} \left((T+ N^{\frac{1}{\alpha}})^{1- \alpha} - N^{\frac{1-\alpha}{\alpha}} \right) & \mbox{for } \alpha \ne 1 \\
\log\left(1 + T/N \right) & \mbox{for } \alpha = 1.
\end{array}\right.
\end{equation}
As $T \to \infty$,
the dominant term in $\frac{1}{1-\alpha} \left((T+ N^{\frac{1}{\alpha}})^{1- \alpha} - N^{\frac{1-\alpha}{\alpha}} \right)$ is
$\frac{1}{\alpha - 1} N^{\frac{1-\alpha}{\alpha}}$ if $\alpha > 1$, and is $\frac{1}{1-\alpha} T^{1 - \alpha}$ if $\alpha < 1$;
and the dominant term in $\log\left(1 + T/N \right)$ is $\log T$.
It then suffices to compare $\overline{\nu} \int_0^T \frac{dt}{N_{\alpha}(t)}$ to $\frac{N-x}{N}$, and the rest of the corollary is immediate.
\end{proof}

\quad This corollary shows a sharp phase transition towards monopoly in terms of the rewarding schemes. 
For $\alpha > 1$ (increasing reward), there is a threshold for $\overline{\nu}$, only above which monopoly may occur, 
and below which the share of stakes increases towards the value on the right side of \eqref{eq:33}.
For $\alpha \le 1$ (constant or decreasing reward), monopoly always occurs. 
Thus, these results have practical implications in the design of the PoS protocol. For instance, if/when 
certain participants have large capacities, adopting a suitable increasing reward scheme will counter the  
effect of concentration. 

\smallskip

\quad Now, returning to the proof of Proposition \ref{thm:31}, 
we use the standard machinery of dynamic programming and the HJB equation. 
Consider the following problem, where $V(t,x)$ is the ``value-to-go'' function,
for $0 \le t \le T$ and $0 \le x \le N(t)$:
\begin{equation}
\label{eq:obj23}
\begin{aligned}
V(t,x):=& \max_{\{\nu(s), s\ge t\}} \,\,   \int_t^{\mathcal{T}}e^{-\beta s} \ \ell(X(s)) ds  + e^{-\beta \mathcal{T}} h(X(\mathcal{T})) \\
& \mbox{ subject to } X'(s) = \nu(s) + \frac{N'(s)}{N(s)} X(s), \, X(t) = x, \nonumber \\
& \qquad \qquad \quad \,  0 \le X(s) \le N(s),\nonumber  \\
& \qquad \qquad \quad  \, |\nu(s)| \le \overline{\nu}. \nonumber 
\end{aligned}
\end{equation}
Clearly, the solution to the above problem concerning $V(t,x)$, for all $t\in[0,T]$ and $x\in [0,N(t)]$, will yield the desired solution 
to $U(x)$ in \eqref{eq:obj22}, since  $U(x) = V(0,x)$.
The following lemma identifies an HJB equation (with terminal and boundary conditions), to which.
$V(t,x)$ is a solution.

\begin{lemma}
\label{lem:HJBvis}
Let $Q: = \{(t,x): 0 \le t < T, \, 0<x<N(t)\}$.
Then $V$ is the (unique) viscosity solution to the following HJB equation:
\begin{equation}
\label{eq:HJB3}
\left\{ \begin{array}{lcl}
\partial_t v + e^{-\beta t} \ell(x) + \frac{x N'(t)}{N(t)} \partial_x v + \sup_{|\nu| \le \overline{\nu}} \{\nu \, \partial_x v\} = 0 \quad (t,x)\in Q, \\
v(T,x) = e^{-\beta T} h(x), \\
v(t,0) = e^{-\beta t} h(0), \,\, v(t, N(t)) = e^{-\beta t} h(N(t)).
\end{array}\right.
\end{equation}
\end{lemma}

\begin{proof}
Write the HJB equation as $\partial_t v + H(t,x, \partial_x v) = 0$, where
\begin{equation*}
H(t,x,p):= e^{-\beta t} \ell(x) + \frac{x N'(t)}{N(t)} p + \sup_{|\nu| \le \overline{\nu}}\{\nu p\}.
\end{equation*}
The fact that $V$ as specified in \eqref{eq:obj23} is a viscosity solution follows a standard dynamic programming argument,
see \cite[Chapter II, Section 7]{FS06}.

Moreover, from the conditions in Assumption \ref{assump:1}, we have, 
\begin{equation}
\label{eq:uniqueLip}
|H(t,x,p) - H(s,y,q)| \le C(|t-s| + |x-y| + |p-q| +  |x-y| |p| + |t-s| |p|),
\end{equation}
for $0 \le s,t \le T$ and $0 \le x,y \le N(t)$,
 and for some $C > 0$.
By \cite[Chapter II, Corollary 9.1]{FS06}, the HJB equation in \eqref{eq:HJB3} has a unique viscosity solution, which then 
must be none other than $V$. 
\end{proof}

\quad What remains is to pin down the term 
$\sup_{|\nu| \le \overline{\nu}} \{\nu \, \partial_x v\}$ in the HJB equation, i.e., to identify the maximizing $\nu$.
Given the intuitive solution that $\nu = \overline{\nu} > 0$ (a ``conjecture,'' so far), 
the HJB equation in \ref{eq:HJB3} is expected to be
\begin{equation}
\label{eq:transp3}
\left\{ \begin{array}{lcl}
\partial_t v + e^{-\beta t} \ell(x) + \left( \overline{\nu} + \frac{x N'(t)}{N(t)}\right) \partial_x v = 0 \,\, \mbox{in } Q, \\
v(T,x) = e^{-\beta T} h(x), \\
v(t,0) = e^{-\beta t} h(0), \,\, v(t, N(t)) = e^{-\beta t} h(N(t)),
\end{array}\right.
\end{equation}
which is a transport equation with variable coefficients. 
Now we solve the transport equation \eqref{eq:transp3} by the method of characteristics.
For $0 \le t \le T$ and $0 \le x \le N(t)$, 
let $\gamma_{t,x}(s)$ be the solution to the following equation:
\begin{equation}
\label{eq:gamma}
\gamma'_{t,x}(s) = \overline{\nu} + \frac{N'(s)}{N(s)} \gamma_{t,x}(s), \quad s>t; \qquad \gamma_{t,x}(t) = x.
\end{equation}
A direct computation yields 
\begin{equation}
\label{eq:gamma1}
\gamma_{t,x}(s) = \overline{\nu} N(s) \int_t^s \frac{du}{N(u)} + \frac{x N(s)}{N(t)}, \quad s \ge t .
\end{equation}
Under the regularity conditions in Assumption \ref{assump:1},
it is standard that (see e.g. \cite{Am08, Golse13}) 
\begin{equation}
\label{eq:vtx3}
v(t,x) = e^{-\beta \mathcal{T}_{t,x}} h(\gamma_{t,x}(\mathcal{T}_{t,x})) + \int_t^{\mathcal{T}_{t,x}} e^{-\beta s} \ell(\gamma_{t,x}(s)) ds,
\end{equation}
where $\mathcal{T}_{t,x}: = \inf\{s> t: \gamma_{t,x}(s) = N(s)\} \wedge T$.
We will next show that $v(t,x)$ given by \eqref{eq:vtx3} indeed solves the HJB equation \eqref{eq:HJB3}, 
which then proves Proposition \ref{thm:31}.

\begin{proof}[Proof of Proposition \ref{thm:31}]
From the expression of $\gamma_{t,x}(s)$ in \eqref{eq:gamma1}, we have
\begin{equation}
\label{eq:310}
\partial_x \gamma_{t,x}(s) = \frac{N(s)}{N(t)} > 0 \quad \mbox{and} \quad \partial_x \mathcal{T}_{t,x} \le 0.
\end{equation}
Note that $\gamma_{t,x}(s)/N(s)$ is increasing in $s$. 
There are two cases.

{\bf Case 1}: If $\gamma_{t,x}(T)/N(T) \le 1$, then
$\mathcal{T}_{t,x} = T$ and hence, $v(t,x) = e^{-\beta T} h(\gamma_{t,x}(T)) + \int_t^T e^{-\beta s} \ell(\gamma_{t,x}(s)) ds$.
By the regularity conditions in Assumption \ref{assump:1}, we get
\begin{equation*}
\partial_x v = e^{-\beta T} \frac{N(T)}{N(t)} h'(\gamma_{t,x}(T)) + \int_t^T e^{-\beta s} \frac{N(s)}{N(t)} \ell'(\gamma_{t,x}(s))ds \ge 0,
\end{equation*}
where the non-negativity follows from the fact that $N(t) > 0$ and $\ell, h$ are increasing.

{\bf Case 2}: If $\gamma_{t,x}(T)/N(T) > 1$,
then $\mathcal{T}_{t,x} < T$, and hence 
$v(t,x) = e^{-\beta \mathcal{T}_{t,x}} h(N(\mathcal{T}_{t,x})) + \int_t^{\mathcal{T}_{t,x}} e^{-\beta s} \ell(\gamma_{t,x}(s)) ds$.
As a result,
\begin{align*}
\partial_x v & = -\beta e^{-\beta \mathcal{T}_{t,x}} (\partial_x \mathcal{T}_{t,x}) h(N(\mathcal{T}_{t,x}))+ e^{-\beta \mathcal{T}_{t,x}} (\partial_x \mathcal{T}_{t,x}) (h \circ N)'(\mathcal{T}_{t,x}) \\
& \quad \quad  + \int_t^{\mathcal{T}_{t,x}} e^{-\beta s} \frac{N(s)}{N(t)}\ell'(\gamma_{t,x}(s)) ds + e^{-\beta \mathcal{T}_{t,x}} (\partial_x \mathcal{T}_{t,x}) \ell(N(\mathcal{T}_{t,x})) \\
& = \int_t^{\mathcal{T}_{t,x}} e^{-\beta s} \frac{N(s)}{N(t)}\ell'(\gamma_{t,x}(s))  ds  \\
& \quad \quad  -e^{-\beta \mathcal{T}_{t,x}} \underbrace{(\partial_x \mathcal{T}_{t,x})}_{\le 0 \mbox{ \small by } \eqref{eq:310}}  \underbrace{(- \ell \circ N - (h \circ N)' + \beta h \circ N)}_{\ge 0 \mbox{ \small by } \eqref{eq:tech}}(\mathcal{T}_{t,x})  \\
& \ge 0.
\end{align*}
So, in both cases, we have $\partial_x v(t,x) \ge 0$.
Thus, $v(t,x)$ defined by \eqref{eq:vtx3} is a classical solution and hence, a viscosity solution to the HJB equation in \eqref{eq:HJB3}.
By Lemma \ref{lem:HJBvis}, we conclude $V(t,x) = v(t,x)$, and the optimal control is $\nu_{*}(s) = \overline{\nu}$ for $s \ge t$.
Specializing to $t = 0$ yields the results in Proposition \ref{thm:31} 
(and $\gamma(t)$ defined by \eqref{eq:gammat} is just $\gamma_{0,x}(t)$).
\end{proof}

\subsection{Main theorem and proof}
\label{sc:main}

We are now ready to present the main result of this section, the optimal solution to $U(x)$ in \eqref{eq:45}
 and hence to $U(x)$ in \eqref{eq:obj12}.
 
\begin{theorem}
\label{thm:41}
Assume that $r \le \beta$, and $\widetilde{P}_{\beta}(t)$ in \eqref{eq:Pbeta} satisfies the Lipschitz condition:
\begin{equation}
\label{eq:LipP}
|\widetilde{P}_{\beta}(t) - \widetilde{P}_{\beta}(s)| \le C|t-s| \quad \mbox{for some } C >0.
\end{equation}
Then, $U(x) = v(0,x)$ where $v(t,x)$ is the unique viscosity solution to the following HJB equation,
where $Q: = \{(t,x): 0 \le t < T, \, 0<x<N(t)\}$:
\begin{equation}
\label{eq:HJB4}
\left\{ \begin{array}{lcl}
\partial_t v + e^{-\beta t} \ell(x) + \frac{x N'(t)}{N(t)} \partial_x v + \sup_{|\nu| \le \overline{\nu}} \{\nu ( \partial_x v - \widetilde{P}_{\beta}(t))\} = 0 \quad \mbox{in } Q , \\
v(T,x) = e^{-\beta T} h(x), \\
v(t,0) = e^{-\beta t} h(0), \,\, v(t, N(t)) = e^{-\beta t} h(N(t)).
\end{array}\right.
\end{equation}
Moreover, the optimal strategy is $b_{*}(t) = 0$ and $\nu_{*}(t) = \nu_{*}(t, X_{*}(t))$ for $0 \le t \le \mathcal{T}_{*}$, 
where $\nu_{*}(t,x)$ achieves the supremum in \eqref{eq:HJB4}, and $X_{*}(t)$ solves 
$X_*'(t) = \nu_*(t, X_*(t)) + \frac{N'(t)}{N(t)}X_*(t)$ with $X_*(0) = x$, and $\mathcal{T}_{*}: = \inf\{t>0: X_*(t) = 0 \mbox{ or } N(t)\} \wedge T$.
\end{theorem}

\begin{proof}
Similar to the dynamic programming/HJB approach that proves Lemma \ref{lem:HJBvis} and Proposition \ref{thm:31} above, 
here we consider
\begin{equation}
\label{eq:obj13}
\begin{aligned}
V_2(t,x) := & \max_{\nu(s)} \,\,   \int_t^{\mathcal{T}} (-\widetilde{P}_\beta(s) \nu(s) + e^{-\beta s} \ell(X(s)) ds + e^{-\beta \mathcal{T}}h(X(\mathcal{T}))\\
& \mbox{ subject to } X'(s) = \nu(s) + \frac{N'(s)}{N(s)} X(s), \, X(t) = x,  \nonumber\\
& \qquad \qquad \quad \,  0 \le X(s) \le N(s), \nonumber\\
& \qquad \qquad \quad  \, |\nu(s)| \le \overline{\nu},\nonumber
\end{aligned}
\end{equation}
so that $U(x) = V_2(0,x)$.
By the same dynamic programming argument as above,  
$V_2$ solves in the viscosity sense the HJB equation in \eqref{eq:HJB4}, which can be expressed 
as $\partial_t v + H(t,x,\partial_x v) = 0$, with
\begin{equation*}
H(t,x,p): = e^{-\beta t} \ell(x) + \frac{x N'(t)}{N(t)} p + \sup_{|\nu| \le \overline{\nu}}\{\nu(p - \widetilde{P}_{\beta}(t))\}.
\end{equation*}
It is readily checked that under Assumption \ref{assump:1} and the Liptschiz condition in \eqref{eq:LipP}, 
the inequality in \eqref{eq:uniqueLip} holds. 
Thus, $V_2$ as identified above is the unique viscosity to the HJB equation in \eqref{eq:HJB4}.
The rest of the theorem is straightforward.
\end{proof}

\quad Comparing the HJB equations in \eqref{eq:HJB3} and in \eqref{eq:HJB4}, we see the nonlinear term 
changes from $\sup_{|\nu| \le \overline{\nu}}\{\nu \partial_x v\}$ in the stake-hoarding problem,
to $\sup_{|\nu| \le \overline{\nu}}\{\nu (\partial_x v - \widetilde{P}_\beta(t))\}$ in the stake-trading problem, the 
latter being the general consumption-investment problem.  
The more general HJB equation in \eqref{eq:HJB4} does not have a closed-form solution,  
and neither does the optimal trading strategy $\nu_{*}(t)$.
This calls for numerical methods; see e.g. \cite{OS91, Sou85}.


\section{Linear and Convex Utilities}
\label{sc:linear}

\subsection{Linear utility}
\label{sc:linear1}

\quad Consider the special case of linear utility, $\ell(x) = \ell x$ and $h(x) = hx$, for some 
given (positive) constants $\ell$ and $h$. 
In this case we can derive a closed-form solution to the HJB equation in \eqref{eq:HJB4}, 
and then derive the optimal strategy $\nu_*(t)$ (in terms of $\widetilde{P}_\beta(t)$).

\quad To start with,  the HJB equation in \eqref{eq:HJB4} now specializes to the following,
with $Q: = \{(t,x): 0 \le t < T, \, 0<x<N(t)\}$
(as before, refer to Lemma \ref{lem:HJBvis}):
\begin{equation}
\label{eq:HJB42}
\left\{ \begin{array}{lcl}
\partial_t v + \ell e^{-\beta t} x + \frac{x N'(t)}{N(t)} \partial_x v + \sup_{|\nu| \le \overline{\nu}} \{\nu ( \partial_x v - \widetilde{P}_{\beta}(t))\} = 0 \quad (t,x)\in Q, \\
v(T,x) = hx, \\
v(t,0) = 0, \, v(t, N(t)) = hN(t).
\end{array}\right.
\end{equation}
For the nonlinear term $\sup_{|\nu| \le \overline{\nu}}\{\nu (\partial_x v - \widetilde{P}_\beta(t))\}$,
we have $\nu_*(t,x) = \overline{\nu}$ if $\partial_x v(t,x) \ge \widetilde{P}_\beta(t)$, 
and $\nu_*(t,x) = \overline{\nu}$ if $\partial_x v(t,x) < \widetilde{P}_\beta(t)$.

\quad Next, presuming that $\partial_x v \ge \widetilde{P}_{\beta}(t)$, and ignoring the boundary conditions,
the HJB equation in \eqref{eq:HJB42} becomes 
\begin{equation*}
\partial_t v + \ell e^{-\beta t}x -\overline{\nu} \widetilde{P}_\beta(t) + \left(\overline{\nu} + \frac{x N'(t)}{N(t)} \right) \partial_x v = 0, \quad v(T,x) = hx,
\end{equation*}
which has the (classical) solution
\begin{equation}
\label{eq:v+}
v^+(t,x) := h e^{-\beta T} \gamma^+_{t,x}(T) +  
\int_t^T  \left[\ell e^{-\beta s} \gamma^+_{t,x}(s) - \overline{\nu} \widetilde{P}_\beta(s)\right] ds,
\end{equation}
where
\begin{equation}
\label{eq:g+}
\gamma^+_{t,x}(s): = \overline{\nu} N(s) \int_t^s \frac{du}{N(u)} + \frac{x N(s)}{N(t)},  \quad s \in [t, T]. 
\end{equation}
Similarly,  presuming that $\partial_x v < \widetilde{P}_\beta(t)$ and neglecting the boundary conditions turns
the HJB equation in \eqref{eq:HJB4} into the following form: 
\begin{equation*}
\partial_t v + \ell e^{-\beta t}x  + \overline{\nu} \widetilde{P}_\beta(t) + \left(-\overline{\nu} + \frac{x N'(t)}{N(t)} \right) \partial_x v = 0, \quad v(T,x) = hx,
\end{equation*}
which has the solution
\begin{equation}
\label{eq:v-}
v^{-}(t,x) := h e^{-\beta T} \gamma_{t,x}^{-}(T) +  \int_t^T \left[\ell e^{-\beta s} \gamma^-_{t,x}(s) + \overline{\nu} \widetilde{P}_\beta(s)\right] ds, 
\end{equation}
where
\begin{equation}
\label{eq:g-}
\gamma^-_{t,x}(s): = -\overline{\nu} N(s) \int_t^s \frac{du}{N(u)} + \frac{x N(s)}{N(t)}, \quad s\in [t , T]. 
\end{equation}

\quad The key observation is that
\begin{equation}
\label{eq:410}
\partial_x v^+(t,x) = \partial_x v^{-}(t,x) = \underbrace{\frac{1}{N(t)} \left(h e^{-\beta T} N(T) + \ell \int_t^T  e^{-\beta s} N(s) ds \right)}_{:= \Psi(t)} ;
\end{equation}
and $\Psi(t)$, notably independent of $x$, 
is decreasing in $t\in [0,T]$:
\begin{eqnarray}
\label{eq:psi}
\Psi(0) = he^{-\beta T} \frac{N(T) }{N} + \frac{\ell}{N} \int_0^T e^{-\beta t} N(t) dt 
\;\downarrow (\ge)\; \Psi(t) \; \downarrow (\ge)\; 
 \Psi(T) = he^{-\beta T}.
 \end{eqnarray}
This suggests that
$\nu_*(t) = \overline{\nu}$ (buy all the time) if $\sup_{[0,T]} \widetilde{P}_\beta(t) \le \Psi(T)$;
and $\nu_*(t) = -\overline{\nu}$ (sell all the time) if $\inf_{[0,T]} \widetilde{P}_\beta(t) \ge \Psi(0)$.
Various other scenarios are also possible, 
such as first buy then sell, or first sell then buy, and so forth.

\quad The following proposition classifies all possible optimal strategies corresponding to 
$\widetilde{P}_\beta(t)$ as specified above, 
which we will comment on later.

\begin{proposition}
\label{coro:class}
Let $\ell(x) = \ell x$ and $h(x) = h x$ with $\ell, h > 0$, and $N(t)$ satisfy Assumption \ref{assump:1} (i).
Assume that $\widetilde{P}_\beta(t)$ satisfies the Lipschitz condition in \eqref{eq:LipP}, and that $ \overline{\nu}$ satisfies
the following:
\begin{equation}
\label{eq:toend}
\overline{\nu} \int_0^T \frac{dt}{N(t)} \le \frac{x}{N} \wedge \frac{N-x}{N}.
\end{equation}
Then, the following results hold:
\begin{itemize}
\item[(i)]
Suppose $\widetilde{P}_\beta(t)$ stays constant, i.e., for all $t\in [0,T]$, 
 $\widetilde{P}(t) =\widetilde{P}(0)= P(0)$. 
\begin{enumerate}[itemsep = 3 pt]
\item[(a)]
If $P(0) \ge \Psi(0)$, 
then $\nu_*(t) = - \overline{\nu}$ for all $0 \le t \le T$. 
That is, the participant sells at all time at full capacity.
\item[(b)]
If $P(0) \le \Psi(T)$, then $\nu_*(t) = \overline{\nu}$.
That is, the participant purchases at all time at full capacity.
\item[(c)]
If $\Psi(T) < P(0) < \Psi(0)$,
then 
\begin{equation*}
\nu_*(t)  = \left\{ \begin{array}{lcl}
\overline{\nu} & \mbox{for } t \le t_0, \\
-\overline{\nu} & \mbox{for } t > t_0,
\end{array}\right.
\end{equation*}
where $t_0$ is the unique point in $[0,T]$ such that $P(0) =\Psi(t_0)$ with $\Psi(t)$ defined in \eqref{eq:410}.
That is, the participant first buys and after some time sells, both at full capacity.
\end{enumerate}
\item[(ii)]
Suppose that $\widetilde{P}_\beta(t)$ is increasing in $t\in [0,T]$.
\begin{enumerate}[itemsep = 3 pt]
\item[(a)]
If $P(0) \ge \Psi(0)$, 
then $\nu_*(t) = - \overline{\nu}$ for all $0 \le t \le T$. 
That is, the participant sells all the time at full capacity.
\item[(b)]
If $\widetilde{P}_\beta(T) \le \Psi(T)$, then $\nu_*(t) = \overline{\nu}$.
That is, the participant purchases all the time at full capacity.
\item[(c)]
If $P(0) < \Psi(0)$ and $\widetilde{P}_\beta(T) > \Psi(T)$,
then 
\begin{equation*}
\nu_*(t)  = \left\{ \begin{array}{lcl}
\overline{\nu} & \mbox{for } t \le t_0, \\
-\overline{\nu} & \mbox{for } t > t_0,
\end{array}\right.
\end{equation*}
where $t_0$ is the unique point of intersection of $\widetilde{P}_\beta(t)$ 
and $\Psi(t)$ on $[0,T]$.
That is, the participant first buys and after some time sells, both at full capacity.
\end{enumerate}
\item[(iii)]
Suppose that $\widetilde{P}_\beta(t)$ is decreasing  in $t\in [0,T]$.
\begin{enumerate}[itemsep = 3 pt]
\item[(a)]
If $P(0) \ge \Psi(0)$, then the participant first sells, and may then buy, etc, always (buy or sell) at full capacity,
according to the crossings of $\widetilde{P}_\beta(t)$ and $\Psi(t)$ in $[0,T]$.
\item[(b)]
If $P(0) < \Psi(0)$, then the participant first buys, and may then sell, etc, always (buy or sell) at full capacity,
according to the crossings of $\widetilde{P}_\beta(t)$ and $\Psi(t)$ in $[0,T]$.
\end{enumerate}
\end{itemize}
\end{proposition}

\begin{proof}
Recall that $X_*(t)$ is the state process (number of stakes) corresponding to the optimal strategy $\nu_*(t)$, 
which, as stipulated in the rest of the proposition, will be equal to either $\overline{\nu}$ or $-\overline{\nu}$.
The condition in \eqref{eq:toend} then ensures that  
$0 \le X_*(t) \le N(t)$ for all $t\in [0, T]$, 
so $\mathcal{T}_* = T$ (i.e., there will no forced early exit).

Thus, it suffices to find the optimal strategy $\nu_*(t)$ from 
 $$\sup_{|\nu| \le \overline{\nu}}\{\nu [\partial_x v - \widetilde{P}_\beta(t)]\}
 =\sup_{|\nu| \le \overline{\nu}}\{\nu [\Psi(t) - \widetilde{P}_\beta(t)]\}.$$
 
(i) and (ii).  Since $\Psi(t)$ is decreasing and $\widetilde{P}(t)$ is either constant or increasing, 
 $\Psi(t)-P(t)$ is decreasing. Hence, we have the following cases (for both (i) and (ii)).
 
(a) If $\widetilde{P}(0) =P(0) \ge \Psi(0)$, then $\widetilde{P}(t) \ge \Psi(t)$
 for all $t\in[0, T]$; hence,
$\nu_*(t) = - \overline{\nu}$, and $U(x) = v^-(0,x)$.

(b) Similarly, if $P(0) \le \Psi(T)$, then $\widetilde{P}(t) \le \Psi(t)$ for all $t\in[0, T]$; hence,
$\nu_*(t) = \overline{\nu}$, and $U(x) = v^+(0,x)$.

(c) Otherwise, there will be a unique point for $\Psi(t)-P(t)$ (which is decreasing in $t$) to cross $0$ from above, and let $t_0\in [0,T]$
denote the crossing point.  
This implies that $\nu_*(t) = \overline{\nu}$ for $t \le t_0$, and $\nu_*(t) =-\overline{\nu}$ for $t > t_0$;
and 
$$U(x) = v^+(0,x) - v^+(t_0, \gamma^+_{0,x}(t_0)) + v^-(t_0, \gamma^+_{0,x}(t_0)).$$
Part (iii) is similarly argued, the only complication is that $\Psi(t)-P(t)$ is now non-monotone, and hence, there 
will be multiple points when it crosses $0$.
\end{proof}

\quad Several remarks are in order. 
First note that the condition in \eqref{eq:toend} is to guarantee the constraint (C2') not activated prior to $T$;
that is, to exclude the possibility of monopoly/dictatorship that will trigger a forced early exit.
This condition may well be removed, 
but then we would expect another condition similar to the one in \eqref{eq:tech}
to guarantee the optimality of a strategy when an early exit occurs.

\quad Second, $\widetilde{P}_\beta(t) = \mathbb{E}\big[ e^{-\beta t} P(t)\big]$
combines  $\beta$, which measures the participant's sensitivity towards risk, with the stake price $P(t)$. 
Thus, the monotone properties of $\widetilde{P}_\beta(t)$, which
classify the three parts (i)-(iii) in Proposition \ref{coro:class}, 
naturally connect to  martingale pricing:
$\widetilde{P}_\beta(t)$ being a constant  in (i) makes the process $e^{-\beta t} P(t)$ a martingale;
whereas $\widetilde{P}_\beta(t)$ increasing or decreasing, respectively in (ii) and (iii),
makes $e^{-\beta t} P(t)$ a sub-martingale or a super-martingale.

\quad On the other hand, the function $\Psi(t)=\partial_x v^+ (t,x) =\partial_x v^-(t,x)$ 
represents the rate of return of the participant's utility (from holding of stakes, $x$);
and interestingly, in the linear utility case, this return rate is independent of $x$ while decreasing in $t$.
Thus, the trading strategy is completely determined by comparing this return rate $\Psi(t)$ with the participant's
risk-adjusted stake price (or, valuation) $\widetilde{P}_\beta(t)$:
if $\Psi(t)\ge ({\rm resp.} <) \widetilde{P}_\beta(t)$, then the participant will buy (resp.\ sell) stakes.

\quad Specifically, following (i) and (ii) of Proposition \ref{coro:class},
for a constant or an increasing $\widetilde{P}_\beta(t)$
(corresponding to a risk-neutral or risk-seeking participant),
there are only three possible optimal strategies:
buy all the time, sell all the time, or first buy then sell.
(The first-buy-then-sell strategy echoes the general investment practice that
an early investment pays off in a later day.) 
See Figure \ref{fig:2} for an illustration.
\begin{figure}[htb]
    \centering
\begin{subfigure}{0.5\textwidth}
  \includegraphics[width=\linewidth]{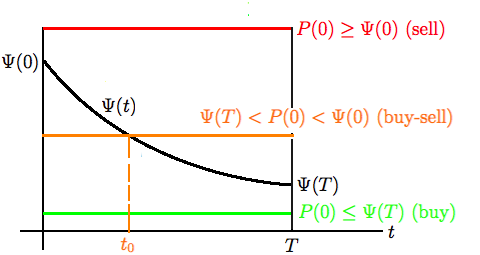}
\end{subfigure}\hfil
\begin{subfigure}{0.5\textwidth}
  \includegraphics[width=\linewidth]{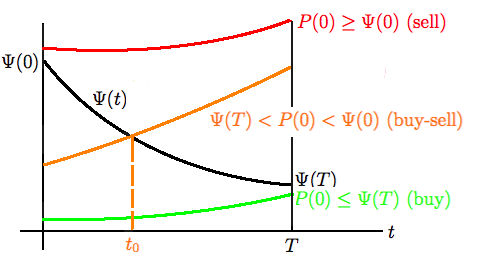}
\end{subfigure}\hfil
\caption{Optimal stake trading with linear $\ell(\cdot), h(\cdot)$ when $\widetilde{P}_\beta(t)$ is constant (left)
and $\widetilde{P}_\beta(t)$ is increasing (right).}
\label{fig:2}
\end{figure}


\subsection{A special case}

In part (iii) of Proposition \ref{coro:class}, when $\widetilde{P}_\beta(t)$ is decreasing in $t$, like $\Psi (t)$,
the multiple crossings between the two decreasing functions can be further pinned down when there's more model structure.
%
Consider, for instance, when $P(t)$ follows a geometric Brownian motion (GBM):
\begin{equation}
\label{eq:gbm}
\frac{dP(t)}{P(t)}=\mu dt +\sigma dB_t, \quad{\rm or}\quad
P(t) = P(0) e^{(\mu - \sigma^2/2) t + \sigma B_t} ;  \quad t \in[0, T],
\end{equation}
where $\{B_t\}$ denotes the standard Brownian motion; and
$\mu> 0$ and  $\sigma > 0$ are the two parameters of the GBM model, representing the rate
of return and the volatility of $\{P(t)\}$.
%
From the second equation in \eqref{eq:gbm}, we have $\mathbb{E} P(t) = P(0) e^{\mu t}$; hence,
 $\widetilde{P}_\beta(t) = P(0) e^{-( \beta-\mu)t}$.
%
%
Then, a decreasing $\widetilde{P}_\beta(t)$ 
corresponds to $\beta > \mu$.
From \eqref{eq:410}, we can derive
\begin{equation*}
\Psi'(t) = -\frac{N'(t)}{N(t)} \Psi (t) - \ell e^{-\beta t}, 
\end{equation*}
and hence, 
\begin{equation}
\label{eq:psider}
\left(\Psi(t) -\widetilde{P}_\beta(t)\right)'= -\frac{N'(t)}{N(t)} \Psi (t) - \ell e^{-\beta t} + (\beta - \mu) P(0) e^{-(\beta - \mu)t}.
\end{equation}
Let $\Psi_\alpha(t)$ denote $\Psi(t)$ for $N(t) = N_\alpha(t)$ defined by \eqref{eq:Nal}.
The following proposition gives the conditions under which $\Psi_\alpha(t) - \widetilde{P}_\beta(t)$ is monotone
in the regime $N \to \infty$,
and optimal strategies are derived accordingly. 

\begin{proposition}
Suppose the assumptions in Proposition \ref{coro:class} hold, 
with $N(t) = N_\alpha(t)$ and $\{P (t)\}$ specified by \eqref{eq:gbm} with $\beta > \mu$.
As $N \to \infty$, we have the following results: 
\begin{itemize}
\item
If for some $\varepsilon > 0$,
\begin{equation}
\label{eq:diffinc}
P(0) > \frac{1}{\beta - \mu} \left(\frac{\alpha h e^{-\mu T} (N^{\frac{1}{\alpha}} +T)^{\alpha}}{N^{1+ \frac{1}{\alpha}}} + \frac{\alpha \ell \beta^{-1}}{N^{\frac{1}{\alpha}}} + \ell \right) + \frac{\varepsilon}{N^{\frac{1}{\alpha}}},
\end{equation}
then $\Psi_\alpha(t) - \widetilde{P}_\beta(t)$ is increasing on $[0,T]$.
\item
If for some $\varepsilon > 0$,
\begin{equation}
\label{eq:diffdec}
P(0) < \frac{1}{\beta - \mu}\left(\frac{\alpha h e^{-\beta T}}{N^{\frac{1}{\alpha}}+T} + \ell e^{-\mu T} \right) - \frac{\varepsilon}{N^{\frac{1}{\alpha}}},
\end{equation}
then 
$\Psi_\alpha(t) - \widetilde{P}_\beta(t)$ is decreasing on $[0,T]$.
\end{itemize}
Consequently, we have:
\begin{enumerate}
\item[(a)]
If $P(0) > e^{(\beta - \mu)T} \Psi_\alpha(T)$ and \eqref{eq:diffinc} holds,
or $P(0) > \Psi_\alpha(0)$ and \eqref{eq:diffdec} holds,
then $\nu_*(t) = - \overline{\nu}$ for all $t$.
That is, the participant sells all the time at full capacity. 
\item[(b)]
If $\Psi_\alpha(0) \le P(0) < e^{(\beta -\mu)T} \Psi_\alpha(T)$ and \eqref{eq:diffinc} holds,
then $\nu_*(t) = - \overline{\nu}$ for $t \le t_0$ and $\nu_*(t) =\overline{\nu}$ for $t > t_0$,
where $t_0$ is the unique point of intersection of $\widetilde{P}_\beta(t)$ 
and $\Psi_\alpha (t)$ on $[0,T]$.
That is, the participant first sells (before $t_0$) and then buys (after $t_0$), both at full capacity.
\item[(c)]
If $e^{(\beta - \mu)T} \Psi_\alpha(T) \le P(0) < \Psi_\alpha(0)$ and \eqref{eq:diffdec} holds,
then $\nu_*(t) = \overline{\nu}$ for $t \le t_0$ and $\nu_*(t) = -\overline{\nu}$ for $t > t_0$,
where $t_0$ is the unique point of intersection of $\widetilde{P}_\beta(t)$ 
and $\Psi_\alpha (t)$ on $[0,T]$.
That is, the participant first buys (before $t_0$) and then sells (after $t_0$), both at full capacity.
\item[(d)]
If $P(0) < e^{(\beta - \mu)T} \Psi_\alpha(T)$ and \eqref{eq:diffdec} holds,
or $P(0) < \Psi_\alpha(0)$ and \eqref{eq:diffinc} holds,
then $\nu_*(t) = \nu$ for all $t$.
That is, the participant buys all the time at full capacity. 
\end{enumerate}
\end{proposition}

\begin{proof}
Note that
$\frac{N'_\alpha(t)}{N_\alpha(t)} = \alpha(N^{\frac{1}{\alpha}} + t)^{-1}$, and 
\begin{equation*}
\int_t^T  e^{-\beta s} N_\alpha(s) ds = 
e^{\beta N^{\frac{1}{\alpha}}} \beta^{-\alpha - 1} \left(\Gamma(\alpha+1, \beta(N^{\frac{1}{\alpha}} + t)) - \Gamma(\alpha+1, \beta(N^{\frac{1}{\alpha}} + T))\right),
\end{equation*}
where $\Gamma(a,x): = \int_x^\infty t^{a-1} e^{-t} dt$ is the incomplete Gamma function.
As $N \to \infty$, we have
\begin{equation*}
\int_t^T  e^{-\beta s} N_\alpha(s) ds = \beta^{-1} \left(e^{-\beta t}N_\alpha(t) - e^{-\beta T} N_\alpha(T)\right) + o(N),
\end{equation*}
which together with \eqref{eq:410} and \eqref{eq:psider} implies that
\begin{equation}
\label{eq:psideral}
\begin{aligned}
\left(\Psi_\alpha(t) -\widetilde{P}_\beta(t)\right)'= -\frac{\alpha}{N^{\frac{1}{\alpha}} +t} &
\left[\frac{h e^{-\beta T} N_\alpha(T)}{N_\alpha(t)} + \ell \beta^{-1} \left(e^{-\beta t} - e^{-\beta T} \frac{N_\alpha(T)}{N_\alpha(t)}\right) + o(1)\right] \\
& - \ell e^{-\beta t} + (\beta - \mu) P(0) e^{-(\beta - \mu)t}.
\end{aligned}
\end{equation}
Multiplying the RHS of \eqref{eq:psideral} by $e^{(\beta - \mu)t}$, we get
\begin{equation*}
\begin{aligned}
-\frac{\alpha}{N^{\frac{1}{\alpha}} +t} &
\left[\frac{h e^{-\beta (t-T) - \mu t} N_\alpha(T)}{N_\alpha(t)} + \ell \beta^{-1} \left(e^{-\mu t} - e^{-\beta (t -T) - \mu t} \frac{N_\alpha(T)}{N_\alpha(t)}\right) + o(1)\right] \\
& - \ell e^{-\mu t} + (\beta - \mu) P(0).
\end{aligned}
\end{equation*}
Clearly, the sum of all the terms above is lower bounded by
\begin{equation*}
-\left(\frac{\alpha h e^{-\mu T} N_\alpha(T)}{N^{1+ \frac{1}{\alpha}}} + \alpha \ell \beta^{-1} N^{-\frac{1}{\alpha}} + \ell \right) + (\beta - \mu) P(0) 
\stackrel{\eqref{eq:diffinc}}{>} 0,
\end{equation*}
which implies that $\inf_{[0,T]} \left(\Psi_\alpha(t) -\widetilde{P}_\beta(t)\right)' > 0$,
and hence, $\Psi_\alpha(t) -\widetilde{P}_\beta(t)$ is increasing.

Moreover, the term is upper bounded by 
\begin{equation*}
-\left(\frac{\alpha h e^{-\beta T}}{N^{\frac{1}{\alpha}}+T} + \ell e^{-\mu T} \right) + (\beta - \mu) P(0) 
\stackrel{\eqref{eq:diffdec}}{<} 0,
\end{equation*}
which implies that $\sup_{[0,T]} \left(\Psi_\alpha(t) -\widetilde{P}_\beta(t)\right)' < 0$,
and hence, $\Psi_\alpha(t) -\widetilde{P}_\beta(t)$ is decreasing.

(a) If $P(0) > e^{(\beta - \mu)T} \Psi_\alpha(T)$ and \eqref{eq:diffinc} holds,
then $\Psi_\alpha(T) < \widetilde{P}_\beta(T)$ and $\Psi_\alpha(t) - \widetilde{P}_\beta(t)$ is increasing.
If $P(0) > \Psi_\alpha(0)$ and \eqref{eq:diffdec} holds,
then $\Psi_\alpha(0) < \widetilde{P}_\beta(0)$ and $\Psi_\alpha(t) - \widetilde{P}_\beta(t)$ is decreasing.
In both cases, we have $\Psi_\alpha(t) - \widetilde{P}_\beta(t) < 0$ for all $t$.

(b) (c) (d) follow the same argument as (a).
\end{proof}

\smallskip

\quad
See Figure \ref{fig:3} below for an illustration of the results in the above proposition.
Also note that the connection to the participant's risk sensitivity as remarked at the end of \S\ref{sc:linear1}
can also be made more explicit when
the price process $P(t)$ follows the GBM model in \eqref{eq:gbm}, for which 
we have  $\widetilde{P}_\beta(t) =  P(0) e^{-(\beta-\mu) t}$. Then, the three cases in Proposition \ref{coro:class} 
correspond to  $\beta =\mu$ (martingale), $\beta <\mu$ (sub-martingale), 
and $\beta >\mu$ (super-martingale). 
According to the three ranges of $\beta$, they can be viewed as representing 
the participant as risk-neutral, risk-seeking and risk-averse.

\begin{figure}[h]
\centering
\includegraphics[width=0.5\columnwidth]{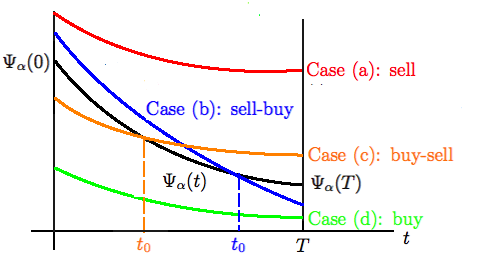}
\caption{Optimal stake trading with linear $\ell(\cdot), h(\cdot)$ when $\widetilde{P}_\beta(t)=P(0) e^{(\mu - \beta)t}$ and $N(t) = N_\alpha(t)$.}
\label{fig:3}
\end{figure}

\subsection{Convex utility}
\label{sc:linear2}

\quad 
It is possible to extend the above results to more general, non-linear utility functions $\ell (\cdot) $ and $h(\cdot)$, by 
following the same approach as above that leads to $v^+(t,x)$ and $v^- (t,x)$ in \eqref{eq:v+} and \eqref{eq:v-}. 
 
\quad Specifically, considering the two cases of $\partial_x v \ge \widetilde{P}_{\beta}(t)$, 
and $\partial_x v < \widetilde{P}_\beta(t)$, we can derive
\begin{eqnarray}
\label{eq:v+1}
&v^+(t,x) := e^{-\beta T} h\big( \gamma^+_{t,x}(T) \big) +  
\int_t^T \left[e^{-\beta s}\ell\big(   \gamma^+_{t,x}(s)\big) - \overline{\nu} \widetilde{P}_\beta(s)\right] ds, \\
&v^{-}(t,x) :=  e^{-\beta T} h\big( \gamma_{t,x}^{-}(T)\big) +  \int_t^T \left[ e^{-\beta s} \ell \big(\gamma^-_{t,x}(s)\big) + \overline{\nu} \widetilde{P}_\beta(s)\right] ds;  \label{eq:v-1} 
\end{eqnarray}
whereas $\gamma_{t,x}^{+}$ and $\gamma_{t,x}^{-}$ remain the same as in \eqref{eq:g+} and \eqref{eq:g-}.

\quad The $\Psi $ function in \eqref{eq:410} now splits into two functions: 
for $(t,x)\in Q: = \{(t,x): 0 \le t < T, \, 0<x<N(t)\}$, we have
\begin{eqnarray}
\partial_x v^+(t,x) = \underbrace{\frac{1}{N(t)} \left(e^{-\beta T}N(T) h'\big( \gamma_{t,x}^{+}(T)\big)  
+  \int_t^T  e^{-\beta s} N(s) \ell' \big(\gamma^+_{t,x}(s)\big)ds \right)}_{:= \Psi^+ (t,x)},  \label{eq:psi+}
\end{eqnarray}
and 
\begin{eqnarray}
\partial_x v^-(t,x) = \underbrace{\frac{1}{N(t)} \left(e^{-\beta T}N(T) h'\big( \gamma_{t,x}^{-}(T)\big)  
+  \int_t^T  e^{-\beta s} N(s) \ell' \big(\gamma^-_{t,x}(s)\big)ds \right)}_{:= \Psi^- (t,x)}. \label{eq:psi-}
\end{eqnarray}



\quad
Note that both $\Psi^+$ and $\Psi^-$ depend on $x$ (as well as on $t$), via $\gamma^+_{t,x}$ and 
$\gamma^-_{t,x}$. This dependence makes it necessary to take a closer look at  $\gamma^+_{t,x}$ and $\gamma^-_{t,x}$, 
since the $x=x(t)$ involved in both depends on the control $\nu$ {\it before} (and up to) $t$.
We have the following cases: for $s\ge t$,
\begin{eqnarray}
{\rm if} \; x=\gamma^+_{0,x}(t), \; {\rm then} &  
\gamma^+_+ (s):=\gamma^+_{t, x}(s) = \left( \overline{\nu} \int_t^s \frac{du}{N(u)} + \overline{\nu} \int_0^t \frac{du}{N(u)} + \frac{x}{N} \right) N(s),
\label{eq:g++}\\
{\rm if} \;x=\gamma^-_{0,x}(t), \; {\rm then} & 
\gamma^-_- (s):=\gamma^-_{t, x}(s) = \left( -\overline{\nu} \int_t^s \frac{du}{N(u)} - \overline{\nu} \int_0^t \frac{du}{N(u)} + \frac{x}{N} \right) N(s) .
\label{eq:g--}
\end{eqnarray}
In other words, $\gamma^+_+$ corresponds to $\nu =\overline{\nu}$ both before and after $t$, whereas 
$\gamma^-_-$ corresponds to $\nu =-\overline{\nu}$ both before and after $t$. 
The other two cases are similar: 
\begin{eqnarray}
{\rm if} \;x =\gamma^+_{0,x}(t), \; {\rm then} & 
\gamma^-_+ (s):=\gamma^-_{t, x}(s) = \left(- \overline{\nu} \int_t^s \frac{du}{N(u)} + \overline{\nu} \int_0^t \frac{du}{N(u)} + \frac{x}{N} \right) N(s),
\label{eq:g+-}\\
{\rm if} \; x=\gamma^-_{0,x}(t), \; {\rm then} & 
\gamma^+_- (s):=\gamma^+_{t, x}(s) = \left( \overline{\nu} \int_t^s \frac{du}{N(u)} -\overline{\nu} \int_0^t \frac{du}{N(u)} + \frac{x}{N} \right) N(s) ;
\label{eq:g-+}
\end{eqnarray}
where $\gamma^-_+$ corresponds to $\nu =\overline{\nu}$ before (and up to) $t$ and $\nu =-\overline{\nu}$ after $t$,
and $\gamma^+_-$ corresponds to the other way around. 

\quad
Substituting these four cases into $\Psi^+$ and $\Psi^-$ in \eqref{eq:psi+} and \eqref{eq:psi-}
further splits the latter two into four cases: 
\begin{eqnarray}
& \Psi^+_+ (t):= \Psi^+ (t, \gamma^+_{0,x}(t)), \quad
\Psi^-_-(t):= \Psi^- (t, \gamma^-_{0,x}(t));
\label{eq:Psi+-1}\\
& \Psi^-_+(t):= \Psi^- (t, \gamma^+_{0,x}(t)), \quad
\Psi^+_-(t):= \Psi^+ (t, \gamma^-_{0,x}(t)).
\label{eq:Psi+-2}
\end{eqnarray}
All four are now 
functions of $t$ only, as $x$ has been replaced by either $\gamma^+_{0,x}(t)$ or $\gamma^-_{0,x}(t)$.

\quad
Clearly, from \eqref{eq:g++}-\eqref{eq:g-+} above, we have  
\begin{eqnarray}
\label{eq:pg+-}
\partial_t \gamma^+_+  (s)= \partial_t \gamma^-_-  (s)=0, 
\quad 
\partial_t \gamma^-_+  (s)= \frac{2 \overline{\nu} N(s)}{N(t)} >0, 
\quad 
\partial_t \gamma^+_-  (s)= -\frac{2 \overline{\nu} N(s)}{N(t)} <0.
\end{eqnarray}

\quad
Now, suppose $\ell (\cdot) $ and $h(\cdot)$ are both smooth, {\it convex} (and increasing) functions.
Hence, $\ell' (\cdot) \ge 0 $ and $h'(\cdot) \ge 0$, and both are increasing functions. 
Then,  it is readily verified:
\begin{itemize}
\item[(i)] Both $\Psi^+_+ (t)$ and $\Psi^-_- (t) $ are decreasing in $t\in [0,T]$, and so is 
$\Psi^+_- (t) $; whereas  $\Psi^-_+ (t) $ could be both increasing and decreasing (i.e., non-monotone). 
\item[(ii)] Furthermore, $\Psi^+_+ (t) \ge \Psi^-_- (t) $ for all $t\in [0,T]$.
\end{itemize}

For instance, for $\Psi^+_+ (t)$ in (i), consider
\begin{eqnarray}
\label{ptpsi}
\partial_t \Psi^+_+ (t) &=& e^{-\beta T}N(T)\left(\frac{h^{''}(\gamma^+_+ (T))\partial_t \gamma^+_+ (T)}{N(t)}
-\frac{h'(\gamma^+_+ (T))N'(t)}{N^2(t)}\right) \nonumber\\
&& -\frac{N'(t)}{N^2(t)} \int_t^T  e^{-\beta s} N(s) \ell' (\gamma^+_+ (s))ds 
-e^{-\beta t} \ell' (\gamma^+_+ (t))\nonumber\\
&&+\frac{1}{N(t)}   \int_t^T  e^{-\beta s} N(s) \ell^{''}(\gamma^+_+ (s)) \partial_t \gamma^+_+ (s) ds 
\quad \le 0, 
\end{eqnarray}
where $\le 0$ follows from $\partial_t \gamma^+_+ (\cdot)= 0$ in both the first and last terms on the RHS.
The other two cases, $\partial_t \Psi^-_- (t) \le 0$ and $\partial_t \Psi^+_- (t) \le 0$, are similarly verified. 

 
 \quad  As in the case of linear utility, the properties above can be used to compare against $\widetilde{P}_\beta(t)$ to 
identify the optimal trading strategy. 
 Consider the case of $\widetilde{P}_\beta(t)$ being a constant, 
 $\widetilde{P}_\beta(t) = P(0)$ for all $t\in [0,T]$, as in part (i) of Proposition \ref{coro:class}. 
 If $\Psi^+_+ (t) \ge \Psi^-_- (t) \ge P(0) $ for all $t\in [0,t]$, then the optimal strategy is to buy all the time 
 and at rate $\overline \nu$.
 If $P(0) \ge \Psi^+_+ (t) >\Psi^-_- (t)$ for all $t\in [0,t]$, then it is optimal to sell all the time, at full capacity.   
 
 \quad
 On the other hand, since $\Psi^+_- $ corresponds to sell first (before $t$) and then buy, this clearly cannot
 be optimal, as it is impossible for $\Psi^+_- \le P(0)$ before $t$ and $\Psi^+_- \ge P(0)$ after $t$, 
 since $\Psi^+_-$ is decreasing in $t$.
 Similarly,  $\Psi^-_+ $ corresponds to buy first (before $t$) and then sell, which can be optimal provided
 if $\Psi^-_+  (t) $ is decreasing in $t$.
 
 
 \quad 
 The details are stated in the following proposition; and
 see Figure \ref{fig:4} for an illustration.
 
 \begin{proposition}
 \label{prop:convex}
 Assume that $\ell(\cdot)$ and $h(\cdot)$ are twice continuously differentiable, convex, and satisfy the conditions in Assumption \ref{assump:1}.
 Assume that $\widetilde{P}_\beta(t)$ stays constant, i.e. $\widetilde{P}_\beta(t) = P(0)$ for all $t \in [0,T]$.
 Further assume the condition \eqref{eq:toend},
 and that $t \to \Psi^-_{+}(t)$ is decreasing
 then
 \begin{enumerate}
 \item[(a)]
 If $P(0) \ge \Psi^+_+(T) \vee \Psi^-(0,x)$,
 then $\nu_*(t) = -\overline{\nu}$ for all $0 \le t \le T$.
 That is, the participant sells at all time at full capacity.
 \item[(b)]
 If $P(0) \le \Psi^-_+(T)$,
 then $\nu_*(t) = \overline{\nu}$ for all $0 \le t \le T$.
 That is, the participant buys at all time at full capacity.
 \item[(c)]
 If $\Psi^+_+(T) < \Psi^-(0,x)$ and $\Psi^-_+(T) < P(0) < \Psi^-(0,x)$,
then 
 \begin{equation*}
\nu_*(t)  = \left\{ \begin{array}{lcl}
\overline{\nu} & \mbox{for } t \le t_0, \\
-\overline{\nu} & \mbox{for } t > t_0,
\end{array}\right.
\end{equation*}
 where $t_0$ is the unique point in $[0,T]$ such that $\Psi^-_{+}(t) = P(0)$.
 That is, the participant first buys and after some time sells, both at full capacity.
 \item[(d)]
 If $\Psi^-(0,x) < \Psi^+_+(T)$, then
 \begin{enumerate}
 \item[(1)]
 if $\Psi^-(0,x)< P(0) < \Psi^+_+(T)$, then $\nu_*(t) = -\overline{\nu}$ for all $0 \le t \le T$.
 That is, the participant sells at all time at full capacity.
 \item[(2)]
 if $\Psi^-_{+}(T) < P(0) \le \Psi^-(0,x)$, then 
 then 
 \begin{equation*}
\nu_*(t)  = \left\{ \begin{array}{lcl}
\overline{\nu} & \mbox{for } t \le t_0, \\
-\overline{\nu} & \mbox{for } t > t_0,
\end{array}\right.
\end{equation*}
 where $t_0$ is the unique point in $[0,T]$ such that $\Psi^-_{+}(t) = P(0)$.
 That is, the participant first buys and after some time sells, both at full capacity.
 \end{enumerate}

 \end{enumerate}
 \end{proposition}
 
 \begin{figure}[htb]
    \centering
\begin{subfigure}{0.5\textwidth}
  \includegraphics[width=\linewidth]{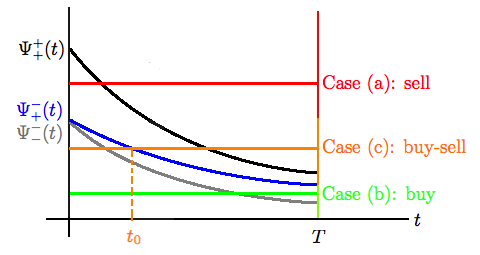}
\end{subfigure}\hfil 
\begin{subfigure}{0.5\textwidth}
  \includegraphics[width=\linewidth]{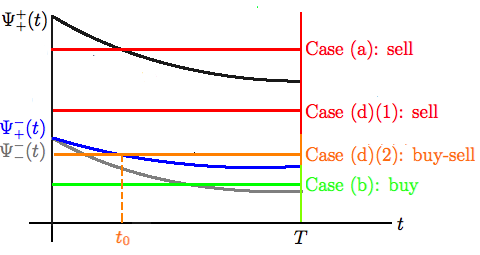}
\end{subfigure}\hfil
\caption{Optimal stake trading with convex $\ell(\cdot), h(\cdot)$ when $\widetilde{P}_\beta(t)$ is constant,
and $\Psi^+_+(T) <\Psi^-(0,x)$ (left) and $\Psi^+_{+}(T)> \Psi^-(0,x)$ (right).}
\label{fig:4}
\end{figure}

\section{Extension: Risk Control}
\label{sc5}



\quad In the previous sections, we have focused on profit seeking objectives
in which a participant's utility increases with getting more stakes, or consuming more.
In the modern finance literature, 
Markowitz \cite{Mar59} pioneered the idea of balancing return and risk in any investment, which is particularly relevant 
for cryptocurrency trading, which often involves substantial volatility.
In this spirit, here we add to the utility objective two ``cost'' terms that penalize the deviation of participant $k$'s
holding of stakes from the average of all others. 
The idea is, to extent this deviation measures risk (analogous to the variance in the Markowitz model), 
it should be the price to be paid for  
the utility (in holding stakes) that $k$ wants to maximize.  
(The same idea has been used in \cite{GTX22} in the context of stochastic games.)
Specifically, the deviation of participant $k$'s holding from the average all others can be expressed as
$| X_k(t) - \frac{N(t)}{K}|$, taking into account $N(t)=\sum_{k =1}^K X_k(t)$. Hence, the new objective function is:
\begin{align}
\label{eq:obj55}
U(x):=\sup_{\{\nu(t), b(t)\}}  & J(\nu, b):= \mathbb{E}\bigg\{ \int_0^{\mathcal{T}}e^{-\beta t} (dc(t) + \ \ell(X(t)) dt) + e^{-\beta \mathcal{T}} \left( b(\mathcal{T}) + h(X(\mathcal{T}) \right)  \notag \\
& \qquad  -  \int_0^{\mathcal{T}} e^{-\delta t}g\left(X(t) - \frac{N(t)}{K}\right) dt - e^{-\delta \mathcal{T} }q\left(X(\mathcal{T}) - \frac{N(\mathcal{T})}{K}\right) \bigg\} \\
& \mbox{ subject to } X'(t) = \nu(t) + \frac{N'(t)}{N(t)} X(t), \, X(0) = x, \tag{C0} \\
& \qquad \qquad \quad \, dc(t) + db(t) - rb(t)dt + P(t) \nu(t) dt = 0, \tag{C1} \\
& \qquad \qquad \quad \,  b(0) = 0, \, b(t) \ge 0 \mbox{ and } 0 \le X(t) \le N(t), \tag{C2} \\ 
& \qquad \qquad \quad  \,  |\nu(t)| \le \overline{\nu}, \tag{C3}
\end{align}
where $\delta > 0$ is a discount factor (which may or may not be equal to $\beta$),
and $g: \mathbb{R} \to \mathbb{R}_{+}$ 
and $q: \mathbb{R} \to \mathbb{R}_{+}$ are symmetric, and increasing on $\mathbb{R}_{+}$
(a typical example is $g(x) = g x^2$ and $q(x) = qx^2 $ with $g, q > 0$).

\quad The theorem below follows the same argument as Theorem \ref{thm:41}.

\begin{theorem}
\label{thm:51}
Let the assumptions in Theorem \ref{thm:41} hold for the problem \eqref{eq:obj55}.
Assume that 
$g, q \in \mathcal{C}^1(\mathbb{R})$
are symmetric, and increasing on $\mathbb{R}_{+}$.
Then $U(x) = v(0,x)$ where $v(t,x)$ is the unique viscosity solution to the following HJB equation:
\begin{equation}
\label{eq:HJB5}
\left\{ \begin{array}{lcl}
\partial_t v + e^{-\beta t} \ell(x) - e^{-\delta t} g\left(x - \frac{N(t)}{K} \right)+ \frac{x N'(t)}{N(t)} \partial_x v + \sup_{|\nu| \le \overline{\nu}} \{\nu ( \partial_x v - \widetilde{P}_{\beta}(t))\} = 0 \quad \mbox{in } Q, \\
v(T,x) = e^{-\beta T}h(x) - e^{-\delta T} q\left( x - \frac{N(T)}{K}\right), \\
v(t,0) = e^{-\beta t} h(0) - e^{-\delta t} q\left(\frac{N(t)}{K}\right), \,\, v(t, N(t)) = e^{-\beta t} h(N(t)) - e^{-\delta t} q\left(\frac{(K-1)N(t)}{K}\right).
\end{array}\right.
\end{equation}
Moreover, the optimal strategy is $b_{*}(t) = 0$ and $\nu_{*}(t) = \nu_{*}(t, X_{*}(t))$ for $0 \le t \le \mathcal{T}_{*}$ (if it exists),
where $\nu_{*}(t,x)$ achieves the supremum in \eqref{eq:HJB4}, and $X_{*}(t)$ solves 
$X_*'(t) = \nu_*(t, X_*(t)) + \frac{N'(t)}{N(t)}X_*(t)$ with $X_*(0) = x$, and $\mathcal{T}_{*}: = \inf\{t>0: X_*(t) = 0 \mbox{ or } N(t)\} \wedge T$.
\end{theorem}

\quad In general, the HJB equation \eqref{eq:HJB5} does not have a closed-form solution even when
$\ell, h$ are linear, and $g, q$ are quadratic. 
Again it requires numerical methods to solve the HJB equation, and then find the optimal strategy $\nu_{*}$.
Nevertheless, there is one exception where the participant is only concerned with the risk entailed by the stakes.
The objective is to solve the stake parity problem:
\begin{align}
\label{eq:obj56}
U(x):=\inf_{\nu(t)}  & J(\nu):= \int_0^{\mathcal{T}} e^{-\delta t}g\left(X(t) - \frac{N(t)}{K}\right) dt + e^{-\delta \mathcal{T} }q\left(X(\mathcal{T}) - \frac{N(\mathcal{T})}{K}\right) \\
& \mbox{ subject to } X'(t) = \nu(t) + \frac{N'(t)}{N(t)} X(t), \, X(0) = x, \tag{C0} \\
& \qquad \qquad \quad \,  b(0) = 0, \, b(t) \ge 0 \mbox{ and } 0 \le X(t) \le N(t), \tag{C2'} \\ 
& \qquad \qquad \quad  \,  |\nu(t)| \le \overline{\nu}. \tag{C3}
\end{align}
Since $g, h$ attain the minimum at $0$, 
if $x \ge N/K$,
then the participant sells at full capacity until hitting the average $N(t)/K$;
if if $x < N/K$,
then the participant purchases at full capacity until hitting the average $N(t)/K$.
We record this simple fact in the following proposition.

\begin{proposition}
\label{prop:spp}
Assume that $g, q \in \mathcal{C}^1(\mathbb{R})$
are symmetric, and increasing on $\mathbb{R}_{+}$ for the stake parity problem \eqref{eq:obj56}.
Let $\gamma_{+}(t)$ be defined by \eqref{eq:gammat}, and 
\begin{equation}
\label{eq:gammat-}
\gamma_{-}(t): = -\overline{\nu} N(t) \int_0^t \frac{ds}{N(s)} + \frac{x N(t)}{N} \quad \mbox{for } 0 \le t \le T,
\end{equation}
and
\begin{equation}
\label{eq:t+-}
t_{\pm}:= \inf \left\{t>0:  \overline{\nu} \int_0^t \frac{ds}{N(s)} =\pm \left(\frac{1}{K} - \frac{x}{N}\right) \right\}.
\end{equation}
Then, the following results hold.
\begin{enumerate}[itemsep = 3 pt]
\item[(i)]
If $x > N\left(\frac{1}{K} + \overline{\nu} \int_0^T \frac{dt}{N(t)} \right)$,
then the optimal strategy is $\nu_*(t) = - \overline{\nu}$ for all $0 \le t \le T$,
and
$U(x) = \int_0^T e^{-\delta t} g\left(\gamma_-(t) - \frac{N(t)}{K}\right) dt+ e^{-\delta T} q\left(\gamma_-(T) - \frac{N(T)}{K} \right)$.
\item[(ii)]
If $\frac{N}{K} < x \le N\left(\frac{1}{K} + \overline{\nu} \int_0^T \frac{dt}{N(t)} \right)$,
then the optimal strategy is 
\begin{equation*}
\nu_*(t)  = \left\{ \begin{array}{lcl}
-\overline{\nu} & \mbox{for } t \le t_{-}, \\
0 & \mbox{for } t > t_{-},
\end{array}\right.
\end{equation*}
and $U(x) = \int_0^{t_-} e^{-\beta t} g\left(\gamma_-(t) - \frac{N(t)}{K}\right) dt+ \frac{g(0)}{\delta}\left(e^{-\delta t_{-}} - e^{-\delta T} \right)
+ e^{-\delta T} q(0)$.
\item[(iii)]
If $N\left(\frac{1}{K} - \overline{\nu} \int_0^T \frac{dt}{N(t)} \right) \le x < \frac{N}{K}$,
then the optimal strategy is 
\begin{equation*}
\nu_*(t)  = \left\{ \begin{array}{lcl}
\overline{\nu} & \mbox{for } t \le t_{-}, \\
0 & \mbox{for } t > t_{-},
\end{array}\right.
\end{equation*}
and $U(x) = \int_0^{t_+} e^{-\beta t} g\left(\gamma_-(t) - \frac{N(t)}{K}\right) dt + \frac{g(0)}{\delta} \left(e^{-\delta t_{-}} - e^{-\delta T} \right) + e^{-\delta T} q(0)$.
\item[(iv)]
If $x < N\left(\frac{1}{K} - \overline{\nu} \int_0^T \frac{dt}{N(t)} \right)$, 
the the optimal strategy is $\nu_{*}(t) = \overline{\nu}$ for all $0 \le t \le T$, and 
$U(x) = \int_0^T e^{-\delta t} g\left(\gamma_+(t) - \frac{N(t)}{K}\right) dt+ e^{-\delta T} q\left(\gamma_+(T) - \frac{N(T)}{K} \right)$.
\end{enumerate}
\end{proposition}

\begin{proof}
(i) If $x > N\left(\frac{1}{K} + \overline{\nu} \int_0^T \frac{dt}{N(t)} \right)$,
we have $\gamma_-(t) > N(t)/K$ for all $0 \le t \le T$.
By a comparison argument, we get 
$X(t) \ge \gamma_-(t)$ for all $0 \le t \le T$ given any feasible strategy $\nu(t)$.
Since $g, q$ are increasing on $\mathbb{R}_{+}$, we obtain
\begin{eqnarray*}
&&\int_0^T e^{-\delta t}g\left(X(t) - \frac{N(t)}{K}\right) dt + e^{-\delta T} q\left(X(T) - \frac{N(T)}{K} \right) \\
&\ge& \int_0^T e^{-\delta t}g\left(\gamma_-(t) - \frac{N(t)}{K}\right) dt + e^{-\delta T} q\left(\gamma_-(T) - \frac{N(T)}{K} \right) ,
\end{eqnarray*}
which yields the desired result.

\quad (ii) If $\frac{N}{K} < x \le N\left(\frac{1}{K} + \overline{\nu} \int_0^T \frac{dt}{N(t)} \right)$,
we have $\gamma_-(t) > N(t)/K$ for $0 \le t < t_{-}$ and $\gamma_-(t_{-}) = N(t_{-})/K$. 
Again by the comparison argument, $X(t) \ge \gamma_{-}(t)$ for $0 \le t \le t_{-}$ given any strategy.
Thus,
\begin{align*}
& \quad \int_0^T e^{-\delta t}g\left(X(t) - \frac{N(t)}{K}\right) dt + e^{-\delta T} q\left(X(T) - \frac{N(T)}{K} \right) \\
& = \int_0^{t_{-}} e^{-\delta t}g\left(X(t) - \frac{N(t)}{K}\right) dt +  \int_{t_{-}}^T e^{-\delta t}g\left(X(t) - \frac{N(t)}{K}\right) dt + e^{-\delta T} q\left(X(T) - \frac{N(T)}{K} \right) \\
& \ge \int_0^{t_-} e^{-\delta t}g\left(\gamma_-(t) - \frac{N(t)}{K}\right) dt + g(0) \int_{t_{-}}^T e^{-\delta t} dt + e^{\delta T} g(0),
\end{align*}
which permits to conclude. 

\quad (iii) and (iv) follow the same argument as (1) and (2).
\end{proof}

\section{Conclusion}
\label{sc6}

\quad
We have developed in this paper a continuous-time control approach to the 
optimal trading under the PoS protocol, formulated as a  consumption-investment problem.
We present general solutions to the optimal control via dynamic programming and the HJB equations, and in the case
of linear and utility functions, close-form solutions in the form of bang-bang controls.  
Furthermore, we bring out the explicit connections between the 
rate of return in trading/holding stakes and the participant's risk-adjusted valuation of the stakes,
such that the participant's risk sensitivity is explicitly accounted for in the trading strategy.   
We have also studied a risk-control version of the consumption-investment problem, 
and for a special case, the ``stake-parity'' problem, 
we show a mean-reverting strategy is the optimal solution.

\quad 
While our focus here is entirely on an individual participant's trading strategy in a PoS protocol,
it is possible to study the interactions among the participants,
and formulate the problem of trading in a PoS protocol as a game (deterministic or stochastic), and to
study issues such as equilibrium, social
welfare, and the inclusion of a trusted third party (or market maker). 
This will be our focus of a follow-up paper.

\bigskip
{\bf Acknowledgement:} 
W.\ Tang gratefully acknowledges financial support through NSF grants DMS-2113779 and DMS-2206038,
and through a start-up grant at Columbia University.
David Yao's work is part of a Columbia-CityU/HK collaborative project that is supported by InnotHK Initiative, The Government of the HKSAR and the AIFT Lab.

\bibliographystyle{abbrv}
\bibliography{unique}
\end{document}